\documentclass[letterpaper, 10 pt, conference]{ieeeconf} 

\usepackage{graphicx}
\usepackage{cite}
\usepackage{amsmath}
\usepackage{amssymb,amsfonts}
\usepackage{textcomp}
\usepackage{xcolor}
\usepackage{subfigure}
\usepackage{enumerate}
\usepackage[hidelinks]{hyperref}

\usepackage{flushend}
\usepackage{pgf}
\usepackage{tikz}
\usetikzlibrary{positioning, arrows.meta,automata}
\usetikzlibrary{calc}
\usepackage[latin1]{inputenc}

\usepackage[ruled,vlined,linesnumbered, procnumbered]{algorithm2e}

\newtheorem{theorem}{\bf Theorem}

\newtheorem{definition}{\bf Definition}
\newtheorem{proposition}{\bf Proposition}
\newtheorem{problem}{\bf Problem}

\IEEEoverridecommandlockouts                              

\title{\bf A Stackelberg Game Approach for Signal Temporal Logic \\ Control Synthesis with Uncontrollable Agents}

\author{Bohan Cui, Xinyi Yu, Alessandro Giua and Xiang Yin% 
\thanks{This work was supported by  the National Natural Science Foundation of China (62061136004, 62173226, 61833012).}
	\thanks{Bohan Cui and Xiang Yin are with the Department of Automation and Key Laboratory of System Control and Information Processing, Shanghai Jiao Tong University, Shanghai 200240, China.
	{\tt\small \{bohan\_cui,  yinxiang\}@sjtu.edu.cn}. 
	Xinyi Yu is with Thomas Lord Department of Computer Science, University of Southern California, Los Angeles, CA 90089, USA.
	{\tt\small xinyi.yu12@usc.edu}.     
	Alessandro Giua is with the Department of Electrical and Electronic Engineering, 
    University of Cagliari, Cagliari 09123, Italy.
	{\tt\small giua@unica.it}.
	} 
}

\begin{document}

\maketitle
\thispagestyle{empty}
\pagestyle{empty}
\setlength{\abovecaptionskip}{0pt}
\setlength{\belowcaptionskip}{3pt}
\setlength{\textfloatsep}{6pt}

%%%%%%%%%%%%%%%%%%%%%%%%%%%%%%%%%%%%%%%%%%%%%%%%%%%%%%%%%%%%%%%%%%%%%%%%%%%%%%%%
\begin{abstract}
In this paper, we investigate the control synthesis problem for Signal Temporal Logic (STL) specifications in the presence of uncontrollable agents. Existing works mainly address this problem in a robust control setting by assuming the uncontrollable agents are adversarial and accounting for the worst-case scenario. While this approach ensures safety, it can be overly conservative in scenarios where uncontrollable agents have their own objectives that are not entirely opposed to the system's goals. Motivated by this limitation, we propose a new framework for STL control synthesis within the \emph{Stackelberg game} setting. Specifically, we assume that the system controller, acting as the leader, first commits to a plan, after which the uncontrollable agents, acting as followers, take a best response based on the committed plan and their own objectives. Our goal is to synthesize a control sequence for the leader such that, for any rational followers producing a best response, the leader's STL task is guaranteed to be satisfied. We present an effective solution to this problem by transforming it into a single-stage optimization problem and leveraging counter-example guided synthesis techniques. We demonstrate that the proposed approach is sound and identify conditions under which it is optimal. Simulation results are also provided to illustrate the effectiveness  of the proposed framework.
\end{abstract}

%%%%%%%%%%%%%%%%%%%%%%%%%%%%%%%%%%%%%%%%%%%%%%%%%%%%%%%%%%%%%%%%%%%%%%%%%%%%%%%%
\section{Introduction}

Decision making and task planning are fundamental challenges in the design of autonomous systems. In recent years, formal methods have gained significant attention for high-level task planning due to their ability to provide rigorous specifications and guarantees for system correctness  \cite{kress2018synthesis,belta2019formal,yin2024formal}.  Among these methods, temporal logic has emerged as a particularly expressive tool, enabling the description of complex temporal and spatial constraints in a mathematically rigorous yet user-friendly manner.  Particularly, Signal Temporal Logic (STL), initially introduced in \cite{maler2004monitoring},  has been extensively studied as it is well-suited for capturing spatial-temporal requirements  real-valued signals in dynamic and uncertain environments \cite{ma2020sastl, ma2020stlnet, silano2021power}.

In the context of control synthesis for STL specifications, one of the most widely used methods is the optimization-based approach \cite{raman2014model, kurtz2022mixed}. This approach encodes the satisfaction of STL formulae using binary variables and transforms the control synthesis problem into a Mixed Integer Linear Programming (MILP) problem. An alternative approach leverages control barrier functions, which provide sufficient conditions to characterize forward invariant sets corresponding to the satisfaction regions of STL formulae \cite{lindemann2018control, xiao2021high}. These methods primarily address the open-loop control problem, where the goal is to synthesize a single control sequence that satisfies the STL specifications without external disturbances.

When external disturbances or uncertainties are present, feedback control based on online information becomes essential. In the context of STL control synthesis, a direct approach is to integrate open-loop trajectory synthesis methods into the framework of receding horizon control; see, e.g.,  \cite{raman2015reactive, lindemann2021reactive, gundana2021event, ilyes2023stochastic, scher2023ensuring, farahani2018shrinking, sadigh2016safe}. This involves solving the optimization problem iteratively in real-time, where only the latest control input is applied at each step. 
Particularly, to ensure the feasibility of recursive computation, it is necessary to solve a robust optimization problem that guarantees the satisfaction of the STL task under all possible disturbances \cite{farahani2018shrinking, ren2022reachability, yu2023model, vlahakis2024probabilistic}. 

The above-mentioned existing works on STL control synthesis in the presence of external inputs essentially fall into the category of \emph{zero-sum games}. In this framework, the controller aims to ensure the satisfaction of the STL formula under \emph{all possible} external inputs, which can be interpreted as a second player whose objective is entirely opposed to the system's STL task. This zero-sum setting is suitable when external inputs are generated by disturbances or adversarial agents. However, in some scenarios, this approach can be overly conservative. For example, consider a two-agent control synthesis problem where only Agent 1 is controllable, and the satisfaction of the STL formula for Agent 1 depends on the joint trajectory of both agents. 
In the zero-sum setting, the controller must adopt a conservative plan to guarantee task satisfaction regardless of Agent 2's actions. In practice, however, Agent 2 may have its own objective, which could be described by another STL formula that is not entirely opposed to Agent 1's goal. In such cases, synthesizing a plan for Agent 1 should take into account the rationality of Agent 2 to reduce conservatism. This essentially requires to consider \emph{non-zero sum game}  setting, where each player has its own distinct objective.

Motivated by the above considerations, this paper proposes a new perspective for STL control synthesis in the presence of uncontrollable agents within the framework of Stackelberg game or leader-follower game. In this setting, the system's inputs are divided into two parts: 
controlled inputs and uncontrollable inputs. Unlike the traditional robust control or zero-sum game approaches, here we assume that both the controller (leader) and the uncontrollable agent (follower) have their own objectives, each specified by separate STL formulae that are not necessarily opposed to one another. Specifically, the leader first commits to a sequence of control inputs, after which the follower determines a best-response control sequence that optimizes its own STL task. Our objective is to synthesize a control sequence for the leader such that, for any rational follower producing a best response, the leader's STL task is guaranteed to be satisfied. We show that this problem can yield two types of solutions: cooperative solutions, where the leader and follower align their objectives to achieve mutual benefit, and antagonistic solutions, where  the leader forces the follower to be non-interfering by eliminating the possibility of the follower achieving its STL task. For each type of solution, we propose an effective approach to synthesize optimal control sequences for the leader. Finally, we provide two case studies on robot path planning to illustrate the effectiveness  of our proposed framework.

\textbf{Related Works: }
Our work is closely related to hierarchical games within the context of dynamic game theory \cite{bacsar1998dynamic, zhu2011stackelberg} and control synthesis involving uncontrollable agents subject to temporal logic tasks \cite{ulusoy2014receding, hoxha2016planning, li2021safe}. Particularly, some existing works have already explored the integration of these two areas. For instance, \cite{niu2020optimal} addresses the minimum violation problem with numeric utilities for stochastic systems within the framework of Stackelberg games, where  Pareto-optimal payoff is used to determine the best response of the follower. In \cite{li2022dynamic}, the authors investigate the use of deception in strategic planning for adversarial environments by introducing a hierarchical hyper-game model. This model enables the leader to develop deceptive strategies to influence the follower's perception, thereby maximizing the likelihood of achieving its own objective. Additionally, \cite{cui2023towards} studies supervisory control of discrete-event systems within the Stackelberg game framework. However, these works primarily focus on either simple safety specifications or linear temporal logic specifications. To the best of our knowledge, the application of Stackelberg games to STL control synthesis remains unexplored.

% There are also several recent works that consider the leader-follower setting in the context of control synthesis of STL, see, e.g., \cite{chen2024cooperative, sharifi2022higher, restrepo2022asymptotic}. 
% However, these works considered the coupling of dynamics among the agents and the cooperative scenarios, while in this paper we use Stackelberg games to describe the relationship in the decision order of different agents and we investigate both the cooperative and antagonistic cases.

The remainder of this paper is organized as follows. Section~\ref{sec-pre} presents some necessary preliminaries. 
In Section~\ref{sec-pro}, we introduce the 
framework of the leader-follower game  and formulate the Stackelberg STL synthesis problem (SSP-STL). 
Section~\ref{sec-two} presents a two-stage synthesis procedure for finding the cooperative and antagonistic solutions to SSP-STL.
In Section~\ref{sec-case}, we conduct two case studies to illustrate the application of our approach. Finally, we conclude the paper and discuss future directions in Section~\ref{sec-conclusion}.

\section{Preliminaries}\label{sec-pre}
\subsection{System Model}
We consider  discrete-time control systems operating in open, interactive environments where external disturbances or adversarial agents are present. In such cases, the state of the system may be influenced by both the system controller and the \emph{external input}. Consequently, the system dynamics are of the following form
\begin{align}\label{sys-dis}
    x_{t+1}=f(x_t,u_t,w_t), 
\end{align}
where 
$t=0,1,\dots$ are the time indices, 
$x_t\in X \subseteq \mathbb{R}^n$ is the system state at time $t$, $u_t\in U\subseteq  \mathbb{R}^m$ denotes the bounded input of the system controller at time $t$
and $w_t\in W \subseteq \mathbb{R}^e$ denotes the external input. 
We assume that the initial state is fixed as $x_0\in X$. 
Given a sequence of control inputs 
$\mathbf{u}_{0:N-1}=u_0 u_{1}\dots u_{N-1}\in U^{N}$ and a sequence of external inputs $\mathbf{w}_{0:N-1}=w_0 w_{1}\dots w_{N-1}\in U^{N}$ with horizon $N$, 
the resulting state sequence of the system is $\xi_f(x_0,\mathbf{u}_{0:N-1},\mathbf{w}_{0:N-1})=\mathbf{x}_{0:N}=x_0 x_1\dots x_N\in X^{N+1}$, where for all $ i=0,1,\dots,N-1$ we have $x_{i+1}=f(x_i,u_i,w_i)$. 
Hereafter, we will drop the subscript $0:N-1$ when the horizon $N$ is clear from the context.

\subsection{Signal Temporal Logic Specifications}
We consider formal specifications described by signal temporal logic (STL) formulae, whose syntax  is  as follows
\[
\phi::= \top \mid \pi^\mu\mid \lnot\phi \mid  \phi_1\land \phi_2\mid \phi_1 \mathbf{U}_{[a,b]} \phi_2,
\]
where 
$\top$ is the true predicate, 
$\pi^\mu$ is an atomic predicate whose truth value is determined by the sign of its underlying predicate function $\mu: \mathbb{R}^n\to \mathbb{R}$, i.e., $\pi^\mu$ is true at time $t$ if and only if  $\mu(x_t)\geq 0$. Notations $\neg$ and $\land$ are the standard Boolean operators ``negation" and ``conjunction", respectively, and 
$\mathbf{U}_{[a,b]}$ is the temporal operator ``until", where $a,b\in \mathbb{N}$ and $a\leq b$.
Based on the basic notations given above, one can further induce Boolean operators ``disjunction" by $\phi_1\vee\phi_2:=\neg(\neg\phi_1\land\neg\phi_2)$ and ``implication" by $\phi_1\to\phi_2:=\neg\phi_1\vee\phi_2$, and temporal operators ``eventually" by $\mathbf{F}_{[a,b]}\phi:= \top \, \mathbf{U}_{[a,b]} \phi$ and ``always" by $\mathbf{G}_{[a,b]}\phi:=\lnot\mathbf{F}_{[a,b]}\lnot\phi$.

Given a state sequence $\mathbf{x}$, 
we denote by $(\mathbf{x},t)\models \phi$ the satisfaction of STL formula $\phi$ at time $t$. Formally, the Boolean semantics of STL  are  defined as recursively follows:
\begin{align}
    & (\mathbf{x},t)\models \pi^\mu &\Leftrightarrow \quad &\mu(x_t)\geq 0 \nonumber \\
    & (\mathbf{x},t)\models\neg \phi &\Leftrightarrow \quad &\neg((\mathbf{x},t)\models \phi) \nonumber \\
    & (\mathbf{x},t)\models \phi_1\land \phi_2 &\Leftrightarrow \quad & (\mathbf{x},t)\models \phi_1 \land (\mathbf{x},t)\models \phi_2 \nonumber \\
    & (\mathbf{x},t)\models \phi_1 \mathbf{U}_{[a,b]} \phi_2 &\Leftrightarrow \quad & \exists t'\in [t+a,t+b]\!:\! (\mathbf{x},t')\models \phi_2 \nonumber \\
    &&&\text{and } \forall t''\in [t,t']: (\mathbf{x},t'')\models \phi_1 \nonumber
\end{align}
We write $\mathbf{x}\models \phi$ instead $(\mathbf{x},0)\models \phi$ for simplicity.

% In addition to Boolean semantics, one can also use the robustness value $\rho^\phi (\mathbf{x},t) \in \mathbb{R}$ to quantify the degree to which $\phi$ is satisfied. Formally, the robust semantics of STL are defined recursively as follows:
% \begin{align}
%     & \rho^\top (\mathbf{x},t) & = \quad & \infty \nonumber \\
%     & \rho^{\pi^\mu} (\mathbf{x},t) & = \quad &\mu(x_t) \nonumber \\
%     & \rho^{\lnot\phi} (\mathbf{x},t) & = \quad & -\rho^{\phi} (\mathbf{x},t) \nonumber \\    
%     & \rho^{\phi_1\land \phi_2} (\mathbf{x},t) &= \quad & \text{min}((\rho^{\phi_1}(\mathbf{x},t), \rho^{\phi_2}(\mathbf{x},t)) \nonumber \\
%     & \rho^{\phi_1 \mathbf{U}_{[a,b]} \phi_2}(\mathbf{x},t) &= \quad & \text{max}_{t'\in [t+a,t+b]} \text{min}(\rho^{\phi_2}(\mathbf{x},t'), \nonumber \\
%     &&&\text{min}_{t''\in [t,t']} \rho^{\phi_1}(\mathbf{x},t''))\nonumber
% \end{align}

In this paper, we consider bounded-time STL formulae, in which time intervals $[a,b]$ are bounded. This horizon determines the length of the sequence required to evaluate the satisfaction of $\phi$.
Although this work mainly addresses the Boolean semantics, in some cases, it is useful to further consider the quantitative semantics, where the robustness value $\rho^\phi (\mathbf{x}) \in \mathbb{R}$ is used to quantify the degree to which $\phi$ is satisfied. 
The reader is referred to \cite{donze2010robust} for more details for the computation of the robustness value.

\section{Problem Formulation}\label{sec-pro}

\subsection{Robust STL Control Synthesis with Disturbances}
In the literature, the problem of \emph{robust STL control synthesis} has been extensively studied. This problem focuses on designing an input sequence from the perspective of the controller that ensures the satisfaction of the STL formula \emph{for all possible} external input sequences.
Formally, this problem can be described by the following optimization problem
\begin{align}
    & \text{minimize}_{\mathbf{u}}\max_{\mathbf{w}} J(x_0, \mathbf{u},\mathbf{w}) \nonumber \\
    &\text{subject to} \quad  \forall \mathbf{w}\in W^N: 
    \xi_f(x_0,\mathbf{u},\mathbf{w})\models\phi, \nonumber
\end{align}
where  $J: X\times U^N \times W^N \to \mathbb{R}$ is a generic cost function such as the robust value of the resulting state sequence.

The robust control synthesis problem described above essentially adopts a \emph{zero-sum game} framework, assuming that the external inputs are entirely adversarial. This assumption is appropriate in scenarios where the external inputs represent disturbances or noise, which are unintentional and lack strategic intent. However, in many practical applications, the external inputs are determined by another \emph{rational agent} that has its own objectives, which may not be entirely opposed to those of the system controller. In such cases, the robust control synthesis formulation may become overly conservative, as it fails to account for the rationality and potential cooperation of the external agent. 

\subsection{Stackelberg Game Formulation for STL Control Synthesis}
In this paper, to better capture the rationality of the external uncontrollable agents, we introduce a new formulation for STL control synthesis within the framework of Stackelberg games (or leader-follower games). Specifically, we assume that the overall input at each time step is jointly determined by two agents, referred to as the \emph{leader} and the \emph{follower}. The system dynamics are formalized as follows:
\begin{align}\label{sys-stack}
    x_{t+1} = f(x_t, u^L_t, u^F_t),
\end{align}
where 
$u^L_t \in U^L \subseteq \mathbb{R}^{m_L}$ and 
$u^F_t \in U^F \subseteq \mathbb{R}^{m_F}$ represent the leader's and the follower's control inputs at time $t$, respectively. 
Similarly, given the initial state $x_0\in X$, and  two sequences of control inputs $\mathbf{u}^L=u^L_0 u^L_{1}\dots u^L_{N-1}\!\!\in \!\!(U^L)^{N}$ and $\mathbf{u}^F=u^F_0 u^F_{1}\dots u^F_{N-1}\!\!\in \!\!(U^F)^{N}$, 
the resulting state sequence is $\xi_f(x_0,\mathbf{u}^L,\mathbf{u}^F)=\mathbf{x}=x_0 x_1\dots x_N\in X^{N+1}$, where for each $ i\leq N-1$, we have $x_{i+1}=f(x_i,u^L_i,u^F_i)$. 

We assume that both the leader and the follower have their own objectives, represented by STL formulae $\phi^L$ and $\phi^F$, respectively. Note that these two formulae are not necessarily in opposition to each other. Furthermore, we assume that the leader has knowledge of the follower's objective formula $\phi^F$. 
Within the framework of the Stackelberg game, the decision-making process of the two agents proceeds as follows: 
\begin{itemize}
  \item 
  First, the leader selects its control input sequence $\mathbf{u}^L$ based on its objective $\phi^L$ and its knowledge of the follower's objective $\phi^F$.  Once the input sequence of the leader is synthesized, the leader commits to executing this sequence. This decision is then published as public information, making it available to the follower.  
  \item  
  Then, the follower determines its own input sequence $\mathbf{u}^F$, optimizing its objective $\phi^F$ based on the leader's committed actions. 
\end{itemize}
This information structure enables the leader to to \emph{anticipate} the follower's reaction and aims to optimize its own objective while considering the follower's likely response. The follower's decisions are then made with full knowledge of the leader's strategy, reflecting the hierarchical nature of the Stackelberg game framework.

Here, we assume that the behavior of the follower is \emph{rational} in the sense that it will not deviate from its best response.  This leads to the following definitions.  
\begin{definition}[Successful Responses]
Suppose that the leader's input sequence is committed to be  $\mathbf{u}^L$. 
We say  follower's input sequence $\mathbf{u}^F$ is a \emph{successful response} (w.r.t.\ to $\mathbf{u}^L$ and $\phi^F$) if 
\[
\xi_f(x_0,\mathbf{u}^L,\mathbf{u}^F)\models\phi^F.
\]
We denote by $\mathbf{SR}(\mathbf{u}^L,{\phi^F})$ the set of successful responses.
\end{definition}

Note that there may be situations where there are no successful response available to the follower, i.e., $\mathbf{SR}(\mathbf{u}^L,\phi^F)=\emptyset$. 
In such cases, for the sake of simplicity, we assume that the follower will choose to be \emph{non-interfering} in the sense that  its input will be consistently zero.  
Our approach can be extended to the case of any pre-defined inputs when the follower cannot achieve its own objective. 

\begin{definition}[Best Responses]\label{def-best-res}
Suppose that the leader's input sequence is committed to be  $\mathbf{u}^L$. 
We say  follower's input sequence $\mathbf{u}^F$ is a \emph{best response} (w.r.t.\ to $\mathbf{u}^L$ and $\phi^F$) if 
(i) it is a successful response when $\mathbf{SR}(\mathbf{u}^L,{\phi^F})\neq \emptyset$; 
or 
(ii) it is the zero sequence $\mathbf{0}=(0^{m_F})^N$ when $\mathbf{SR}(\mathbf{u}^L,{\phi^F})= \emptyset$. 
We denote by $\mathbf{BR}(\mathbf{u}^L,\phi^F)$ the set of best responses.
\end{definition}

Similar to the robust STL synthesis problem, we also consider a generic cost function 
$J_S: X^N \times (U^L)^N \times (U^F)^N \to \mathbb{R}$, whose value is jointly determined by both the leader's and the follower's control inputs. The control objective is to design an input sequence $\mathbf{u}^L$ from the leader's perspective such that the system trajectory satisfies $\phi^L$ under any best response of the follower, while minimizing the worst-case cost. This formulation leads to the following Stackelberg STL Synthesis Problem (SSP-STL):

\begin{problem}[SSP-STL]\label{SSP-STL}
Given the system in Equation~\eqref{sys-stack}, the cost function $J_S$, the initial state $x_0$, the leader's specification $\phi^L$, and the follower's specification $\phi^F$, synthesize an optimal input sequence for the leader, $\mathbf{u}^L$, to minimize the cost function subject to the constraint that the follower takes a best response. Formally, we have
\[
\begin{aligned}
& \underset{\mathbf{u}^L}{\text{minimize}} & & \max_{\mathbf{u}^F \in \mathbf{BR}(\mathbf{u}^L, \phi^F)} J_S(\mathbf{x}, \mathbf{u}^L, \mathbf{u}^F) \\
& \text{subject to} & & \mathbf{x} \models \phi^L \quad \text{for all} \quad \mathbf{u}^F \in \mathbf{BR}(\mathbf{u}^L, \phi^F), \\
& & & x_{t+1} = f(x_t, u^L_t, u^F_t), \quad t = 0, 1, \dots, N-1.
\end{aligned}
\]\medskip
\end{problem}

Compared to the robust STL control synthesis problem, the main challenge in the Stackelberg control synthesis problem is that the set of best responses cannot be determined by a single STL formula alone. Instead, the set of best responses must be constructed simultaneously while synthesizing $\mathbf{u}^L$. This results in a two-stage optimization problem where the leader's and follower's decisions are inherently coupled.
Particularly, the solution to the SSP-STL can be categorized into two types:  
\begin{itemize}  
\item 
Cooperative solution: 
The leader synthesizes an input sequence $\mathbf{u}^L$ that enables the follower to achieve its own objective, i.e., $\mathbf{SR}(\mathbf{u}^L, \phi^F)\neq \emptyset$;  
\item 
Antagonistic solution: The leader synthesizes an input sequence $\mathbf{u}^L$ that forces the follower to be non-interfering, i.e., $\mathbf{SR}(\mathbf{u}^L, \phi^F)=\emptyset$.  
\end{itemize}

% \begin{remark}
% In some cases, it is possible that the alternative response may not be non-interfering. For instance, if a greedy follower chooses to accomplish as many sub-tasks of the original formula $\phi^F$ when there is no successful response, it could lead to interference. It's important to note that the solution to these cases is technically the same as in our non-interfering setting since the best response $\mathbf{u}^F$ can always be determined given a fixed $\mathbf{u}^L$. For simplicity, we will mainly focus on the non-interfering setting in the following sections.
% \end{remark}
\section{Synthesis of Solutions: Two-Stage Approach}\label{sec-two}

In this section, we present solutions to the SSP-STL for both cooperative and antagonistic cases. Our approach is primarily designed for the cooperative case, which is more challenging due to the need to align the objectives of the leader and the follower. We then demonstrate that the solution approach for the cooperative case can be easily adapted to handle the antagonistic case, where the leader aims to enforce the unsatisfiability of the follower's tasks.

\subsection{Synthesis of Cooperative Solution}
We first seek to find a cooperative solution, i.e., to solve the following optimization problem
\begin{subequations}\label{equation-5}
\begin{align}
& \underset{\mathbf{u}^L}{\text{minimize}} & & \max_{\mathbf{u}^F \in \mathbf{SR}(\mathbf{u}^L, \phi^F)} J_S(\mathbf{x}, \mathbf{u}^L, \mathbf{u}^F) \\
& \text{subject to} & & \mathbf{SR}(\mathbf{u}^L, \phi^F) \ne \emptyset, \label{cons-5b} \\
& & & \mathbf{x} \models \phi^L \quad \text{for all} \quad \mathbf{u}^F \in \mathbf{SR}(\mathbf{u}^L, \phi^F), \label{cons-5c} \\
& & &x_{t+1} = f(x_t, u^L_t, u^F_t), \quad t = 0, 1, \dots, N-1.\nonumber
\end{align}  
\end{subequations}

To solve the above optimization problem, motivated by the standard optimization-based approach for STL control synthesis \cite{raman2014model, kurtz2022mixed}, 
we first encode the satisfaction status of the two STL formulae using binary variables $z^L,z^F\in \{0,1\}$ such that
(i) $z^L=1 \Leftrightarrow \mathbf{x}\models \phi^L$, and
(ii) $z^F=1 \Leftrightarrow \mathbf{x}\models \phi^F$.  
Specifically,  the values of  $z^L$ and $z^F$ are determined recursively as follows.  
For each predicate $\mu$ and each time instant $t=0,1,...,N$, one needs to  introduce a binary variable  $z^\mu_t\in \{0,1\}$ such that
\begin{subequations}
    \begin{align}
        & \mu(x_t)\leq M z^\mu_t-\epsilon\nonumber\\
        -& \mu(x_t)\leq M (1-z^\mu_t)-\epsilon, \nonumber
    \end{align}
\end{subequations}
where $M$ and $\epsilon$ are sufficiently large positive numbers and sufficiently small positive numbers, respectively.  These two constraints enforce that $z^\mu_t=1$ if and only if $\mu(\mathbf{x}_t)>0$. 
Then for each sub-formula  of $\phi^L$ or $\phi^F$, one needs to introduce additional binary variables to encode its satisfaction, along with new constraints defined by the Boolean and temporal operations. 
For the sake of simplicity, we omit the intermediate variables and constraints in this discussion; readers are referred to \cite{raman2014model} for the  detailed  encoding process.
Instead, 
we also denote $z^L(x_0,\mathbf{u}^L,\mathbf{u}^F)$ and $z^F(x_0,\mathbf{u}^L,\mathbf{u}^F)$ as the value of binary variables $z^L$ and $z^F$ upon state sequence $\xi_f(x_0,\mathbf{u}^L,\mathbf{u}^F)$.

To handle the \emph{inner} maximization of $J_S$, we introduce an auxiliary variable $k\in \mathbb{R}$ and rewrite Equation~(\ref{equation-5}) as follows\vspace{-15pt}
\begin{subequations}\label{equation-6}
\begin{align}
& \underset{k,\mathbf{u}^L}{\text{minimize}} & & k \\
& \text{subject to} & & z^F(x_0,\mathbf{u}^L,\mathbf{u}^F)=1, \label{cons-6b} \\
& & & \!\!\!\!\! 
J_S(\mathbf{x}', \mathbf{u}^L, \mathbf{u}'^F) \leq k,\forall \mathbf{u}'^F \!\in \mathbf{SR}(\mathbf{u}^L, \phi^F), \label{cons-6c} \\
& & & 
\!\!\!\!\!z^L(x_0,\mathbf{u}^L,\mathbf{u}'^F)=1, \forall \mathbf{u}'^F \!\in \mathbf{SR}(\mathbf{u}^L, \phi^F), \label{cons-6d} \\
& & & x'_0=x_0, \nonumber \\
& & & x_{t+1} = f(x_t, u^L_t, u^F_t), \quad t = 0, 1, \dots, N-1, \nonumber \\
& & & x'_{t+1} = f(x'_t, u^L_t, u'^F_t), \quad t = 0, 1, \dots, N-1. \nonumber
\end{align}  
\end{subequations}

Intuitively, constraint~(\ref{cons-6b}) encodes the requirement in constraint~(\ref{cons-5b}) that $\mathbf{SR}(\mathbf{u}^L, \phi^F) \ne \emptyset$. Constraints~(\ref{cons-6c}) require $k$ to be the upper bound of the inner maximization problem, while constraints~(\ref{cons-6d}) encode constraints~(\ref{cons-5c}). 
By introducing these constraints, we reduce the bi-level optimization problem in Equation~(\ref{equation-5}) to single-level. Formally, we have the following result.
\begin{proposition}
The optimization problems Equation~(\ref{equation-5}) and Equation~(\ref{equation-6}) are equivalent.
\end{proposition}

\begin{proof}
    To prove this proposition, we note that the inner optimization problem is essentially an uncertainty of the objective function. To solve this kind of problem, the introduction of slack variable $k$ and the constraint (\ref{cons-6c}) has been completely investigated and become a standard method \cite{sim2004robust}. 
    Moreover, by encoding STL formulae $\phi^L$ and $\phi^F$, constraint \ref{cons-6b} is equavilent to \ref{cons-5b}, and constraints \ref{cons-6d} is equavilent to \ref{cons-5c}.
    Therefore, Equation~(\ref{equation-6}) has the same solution as Equation~(\ref{equation-5}). 
\end{proof}

Now we can solve Equation~(\ref{equation-5}) by instead solving the optimization problem in Equation~(\ref{equation-6}), where constraints in (\ref{cons-6b}) and (\ref{cons-6d}) are binary constraints and can be handled efficiently.
However, set $\mathbf{SR}(\mathbf{u}^L, \phi^F)$ is not necessarily the same for different $\mathbf{u}^L$, which makes directly verifying constraints~(\ref{cons-6c}) and 
(\ref{cons-6d}) difficult. This is because that one always needs to reconstruct $\mathbf{SR}(\mathbf{u}^L, \phi^F)$ when $\mathbf{u}^L$ changes.
To tackle this challenge, 
we take a conservative yet efficient approach by gradually extending the 
region of  $\mathbf{SR}(\mathbf{u}^L, \phi^F)$.
Specifically, for constraint~(\ref{cons-6d}), we can directly require that $z^L(x_0,\mathbf{u}^L,\mathbf{u}'^F)\geq z^F(x_0,\mathbf{u}^L,\mathbf{u}'^F)$ for all $\mathbf{u}'^F\in (U^F)^N$ since $\mathbf{u}'^F\in \mathbf{SR}(\mathbf{u}^L, \phi^F)$ can always be captured by $z^F(x_0,\mathbf{u}^L,\mathbf{u}'^F)=1$ according to the definitions of $z^F$ and $\mathbf{SR}(\mathbf{u}^L, \phi^F)$.
However, in general, constraints~(\ref{cons-6c}) are nonlinear. To handle these constraints, we further introduce a new variable $\rho^K\in \mathbb{R}$ such that
\[
\rho^K\geq 0 \Leftrightarrow [J_S(\mathbf{x}, \mathbf{u}^L,\mathbf{u}^F)\leq k]\vee[\mathbf{u}^F\not\in \mathbf{SR}(\mathbf{u}^L, \phi^F)].
\]
The above requirement can be captured by the following set of constraints:
\begin{subequations}
    \begin{align}
        & \rho^K\geq k-J_S(\mathbf{x}, \mathbf{u}^L,\mathbf{u}^F), \nonumber \\
        & \rho^K\geq -\rho^{\phi^F} (\mathbf{x},0), \nonumber \\
        & k\!-\!J_S(\mathbf{x}, \mathbf{u}^L,\mathbf{u}^F)-bM \leq\! \rho^K \!\leq k\!-\!J_S(\mathbf{x}, \mathbf{u}^L,\mathbf{u}^F)+bM, \nonumber \\
        & -\rho^{\phi^F}(\mathbf{x},0)-(1-b)M\leq \rho^K \leq -\rho^{\phi^F}(\mathbf{x},0)+(1-b)M, \nonumber
    \end{align}
\end{subequations}
where $b\in \{0,1\}$ is a binary variable, $M$ is a sufficiently large positive constant, $\rho^{\phi^F}(\mathbf{x},0)$ is the robustness value of $\phi^F$ along sequence $\mathbf{x}$ at time instant $0$ \cite{donze2010robust}. Particularly, for $\mathbf{x}=\xi_f(x_0,\mathbf{u}^L,\mathbf{u}^F)$ we have $\rho^{\phi^F}(\mathbf{x},0)\geq 0$ if and only if $\mathbf{u}^F\in \mathbf{SR}(\mathbf{u}^L, \phi^F)$.
Here, the first two constraints enforce that $\rho^K\geq\max[k-J_S(\mathbf{x}, \mathbf{u}^L,\mathbf{u}^F), -\rho^{\phi^F}(\mathbf{x},0)]$, and the last two constraints enforce that $\rho^K$ always equals to either $k-J_S(\mathbf{x}, \mathbf{u}^L,\mathbf{u}^F)$ or $ -\rho^{\phi^F}(\mathbf{x},0)$. 
Therefore, these constraints together enforce that
\begin{equation}
\rho^K=\max[k-J_S(\mathbf{x}, \mathbf{u}^L,\mathbf{u}^F), -\rho^{\phi^F}(\mathbf{x},0)].
\end{equation}
Analogously, we denote $\rho^K(k,x_0,\mathbf{u}^L,\mathbf{u}^F)$ as the value of $\rho^K$ upon $k$ and $\mathbf{x}=\xi_f(x_0,\mathbf{u}^L,\mathbf{u}^F)$.

We finally rewrite the optimization problem as follows
\begin{subequations}\label{equation-7}
\begin{align}
& \underset{k,\mathbf{u}^L}{\text{minimize}} & & k \\
& \text{subject to} & & z^F(x_0,\mathbf{u}^L,\mathbf{u}^F)=1, \label{cons-7b} \\
& & & \rho^K(k,x_0,\mathbf{u}^L,\mathbf{u}'^F)\geq0,\forall \mathbf{u}'^F\! \in\! (U^F)^N, \label{cons-7c} \\
& & & \!\!\!\!\!\!\!\!\!\!\!\!\!\!\!\!\!\!\!\!\!\!\!\!\!\!\!\!
z^L(x_0,\mathbf{u}^L,\mathbf{u}'^F)\geq z^F(x_0,\mathbf{u}^L,\mathbf{u}'^F),\forall\mathbf{u}'^F \!\in\! (U^F)^N, \label{cons-7d} \\
& & & x'_0=x_0, \nonumber \\
& & & x_{t+1} = f(x_t, u^L_t, u^F_t), \quad t = 0, 1, \dots, N-1, \nonumber \\
& & & x'_{t+1} = f(x'_t, u^L_t, u'^F_t), \quad t = 0, 1, \dots, N-1. \nonumber
\end{align}    
\end{subequations}

And we have the following direct result.
\begin{proposition}
The optimization problems Equation~(\ref{equation-6}) and Equation~(\ref{equation-7}) are equivalent.
\end{proposition}
\begin{proof}
    By introducing $\rho^K$, we can conclude that constraints~(\ref{cons-7c}) are equivalent to (\ref{cons-6c}). Moreover, according to the encoding of the formula $\phi^F\to \phi^L$, we have constraints~(\ref{cons-7d}) are equivalent to (\ref{cons-6d}).
    Therefore, Equation~(\ref{equation-7}) has the same solution as Equation~(\ref{equation-6}).
\end{proof}

Equation~(\ref{equation-7}) further reduces the need to reconstruct the successful response set $\mathbf{SR}(\mathbf{u}^L, \phi^F)$ compared to Equation~(\ref{equation-6}). However, constraints~(\ref{cons-7c}) and (\ref{cons-7d}) are defined on the whole action space $(U^F)^N$ of the follower, which can be prohibitively large or even infinite and computationally intractable to optimize globally.
In practice, rather than seeking a globally optimal solution, we typically prioritize finding an effective one for such a complex optimization problem. To achieve this efficiently, we adapt a counterexample-guided inductive synthesis scheme, as shown in Procedure~\ref{algorithm-synthesis-coop}. Specifically, in line~1, we randomly initialize a set of candidate inputs of the follower. Then we compute the optimal solution $\mathbf{u}^L$ and $k$ w.r.t. $\mathbf{U}^{cand}$ in line~3. 
Next in line~4, we find the counterexample that falsify the global constraints $z^L(x_0,\mathbf{u}^L,\mathbf{u}'^F)\geq z^F(x_0,\mathbf{u}^L,\mathbf{u}'^F)$ or $\rho^K(k,x_0,\mathbf{u}^L,\mathbf{u}'^F)\geq0$ when the above $\mathbf{u}^L$ is used. 
For simplicity, we directly write $\rho^K(k,x_0,\mathbf{u}^L,\mathbf{u}^F)$ and $\rho^{\phi^F\to\phi^L}(x_0,\mathbf{u}^L,\mathbf{u}^F)$ as $\rho^K$ and $\rho^{\phi^F\to\phi^L}$, respectively.
If it does not exist, we can return $\mathbf{u}^L$ as the result in line~6. Otherwise, we add the  counterexample to $\mathbf{U}^{cand}$ in line~8 and consider them in the next iteration.

\begin{procedure}
    \caption{Cooperative Synthesis()}
    \label{algorithm-synthesis-coop}
	\KwIn{system dynamic function $f$, cost function $J_S$, initial state $x_0$ and specifications $\phi^L$ and $\phi^F$} 
	\KwOut{$\mathbf{u}^L$}
	Let $\mathbf{U}^{cand}=\{\text{a set of randomly generated }\mathbf{u}^F\}$\\
	\While{True}
        { \text{ } \vspace{-19pt} 
	\begin{align}
            k, \mathbf{u}^L \leftarrow &\text{the solution of Equation~(\ref{equation-7})} \nonumber \\ 
             & \text{by changing } (U^F)^N \text{ for } \mathbf{U}^{cand} \nonumber 
        \end{align} \\
        \text{ } \vspace{-20pt}
        \begin{align}
            &\mathbf{u}^F \leftarrow \underset{\mathbf{u}^F}{\text{minimize}}\min[\rho^K, \rho^{\phi^F\to\phi^L}] \nonumber
        \end{align} \\
	\eIf{$\min[\rho^K,\rho^{\phi^F\to\phi^L}]\geq 0$}
	{\textbf{return} $\mathbf{u}^L$}
	{$\mathbf{U}^{cand}\leftarrow \mathbf{U}^{cand}\cup \{\mathbf{u}^F\}$} 
 }
\end{procedure}

\begin{theorem}\label{theo-1}
    Procedure~\ref{algorithm-synthesis-coop} is sound, in the sense that if it returns $\mathbf{u}^L\in (U^L)^N$, then constraints in (\ref{equation-7}) are all fulfilled.
\end{theorem}

\begin{proof}
    We first note that for any $\mathbf{u}'^F\in \mathbf{U}^{cand}$, we have $\mathbf{u}'^F\in (U^F)^N$. Therefore, if $\mathbf{u}^L$ is returned, then we can first conclude that there exists a $\mathbf{u}^F\in (U^F)^N$, such that $z^F=1$, which immediately leads that constraint (\ref{cons-7b}) is satisfied.
    To prove constraints (\ref{cons-7c}) and (\ref{cons-7d}), we suppose that there still exists $\mathbf{u}^F\in (U^F)^N$ such that $(\rho^K<0)\vee(z^L<z^F)$ holds true
    after $\mathbf{u}^L$ is returned.
    Then we can conclude that $(\rho^K<0)\vee \rho^{\phi^F\to\phi^L}<0$. And thus
    $\min[\rho^K,\rho^{\phi^F\to\phi^L}]<0$, which is a contradiction to the condition in line~5. Therefore, both constraints (\ref{cons-7c}) and (\ref{cons-7d}) are satisfied.
\end{proof}

Note the optimality of the solution returned by Procedure~\ref{algorithm-synthesis-coop} is highly dependent on the initialization of $\mathbf{U}^{cand}$. To see this, let's consider an extreme example, assume that $\mathbf{U}^{cand}$ is initialized as a single instance of $\mathbf{u}^F$, with only a single $\mathbf{u}^L$ as a control inputs pair that meets the condition $z^F(x_0,\mathbf{u}^L,\mathbf{u}^F)=1$, as well as satisfies (\ref{cons-7c}) and (\ref{cons-7d}). In such a situation, the procedure will stop and return this $\mathbf{u}^L$ after one iteration, as no counterexamples can be found. Nonetheless, we cannot guarantee the optimality of this solution. 
In practice, the more likely optimal solution can be constructed by enlarging the initial $\mathbf{U}^{cand}$ or running this procedure repeatedly.

\subsection{Synthesis of Antagonistic Solution}

In the antagonistic case, the leader must enforce the unsatisfiability of the follower's task. By assuming that the follower will become non-interfering when its own task is not satisfiable, the leader can set $\mathbf{u}^F$ to $\mathbf{0}$ while maximizing its objective function. This leads to the following optimization problem.
\begin{subequations}\label{equation-8}
\begin{align}
& \underset{\mathbf{u}^L}{\text{minimize}} & &  J_S(\mathbf{x}, \mathbf{u}^L, \mathbf{0}) \\
& \text{subject to} & & \mathbf{SR}(\mathbf{u}^L, \phi^F) = \emptyset, \label{cons-8b} \\
& & & \mathbf{x} \models \phi^L, \label{cons-8c} \\
& & & x_{t+1} = f(x_t, u^L_t, 0^{m_F}), \quad t = 0, 1, \dots, N-1.\nonumber
\end{align}   
\end{subequations}

To solve this, note that constraints~(\ref{cons-8b}) and (\ref{cons-8c}) can be directly encoded by $z^F(x_0,\mathbf{u}^L,\mathbf{u}'^F)=0,  \forall  \mathbf{u}'^F \in (U^F)^N$ and $z^L(x_0,\mathbf{u}^L,\mathbf{0})=1$.
Therefore, we can still use the basic outline of Procedure~\ref{algorithm-synthesis-coop} to compute the leader's control input sequence $\mathbf{u}^L$ for the antagonistic case by incorporating the following modifications:
\begin{itemize}
    \item 
    In line~3,  Equation~(\ref{equation-7}) is changed to  Equation~(\ref{equation-8}) with encoded constraints~(\ref{cons-8b}) and (\ref{cons-8c}); and
    \item 
    In line~4 and line~5, the condition for determining counterexamples is changed to $\max_{\mathbf{u}^F}\rho^{\phi^F}< 0$.
\end{itemize}

The modified procedure still inherits the  counterexample-guided scheme. However, the key distinction lies in the fact that the initialized set $\mathbf{U}^{cand}$ does not affect the result returned by the procedure since the objective function is always optimized under $\mathbf{u}^F=\mathbf{0}$.
As a result, given that we can get the optimal solution for the problem in line~3, the modified procedure is not only sound but also optimal for the case of the antagonistic solution. 
\begin{theorem}
    In the case of antagonistic, if we can find the optimal solution in line~3, then we have that Procedure~\ref{algorithm-synthesis-coop} is both sound and optimal.
\end{theorem}
\begin{proof}
    The proof of soundness can be processed the same as \ref{theo-1}. As for completeness, we want to prove that the control input $\mathbf{u}^L$ returned by Procedure~\ref{algorithm-synthesis-coop} is indeed the optimal solution in the antagonistic case.
    To prove this, we assume that after $\mathbf{u}^L$ is returned, there still exists a control input $\mathbf{u}'^{L}$ satisfying constraints (\ref{cons-8b}) and (\ref{cons-8c}) such that $J_S(x_0, \mathbf{u}'^{L},\mathbf{0})<J_S(x_0, \mathbf{u}^L,\mathbf{0})$. From this assumption, we know that for any $\mathbf{u}^F$ in the current candidate set $\mathbf{U}^{cand}$, we also have $z^F=0$ under the control input $\mathbf{u}'^{L}$. Therefore, we have $\mathbf{u}^L=\mathbf{u}^{L'}$ in line~3, which is a contradiction to the fact that $\mathbf{u}^L$ is returned.
\end{proof}

\section{Case Studies}\label{sec-case}
In this section, we present simulation results for two case studies. Our approach are implemented in \textsf{Python 3}, and we use \textsf{Gurobi} to solve the optimization problems. All simulations were conducted on a laptop equipped with an Apple M2 CPU and 8 GB of RAM. All source codes are available at 
\href{https://github.com/stack-stl/Stackelberg-STL}{https://github.com/stack-stl/Stackelberg-STL}.

\subsection{Double Integrator with Joint Inputs}
{\bf System Model: }
In this case study, we consider a two-dimensional double integrator system operating in a planar environment, where the acceleration in each dimension is jointly determined by both the leader and the follower. 
The system dynamic,  with
a sampling period of 0.5s, is given by
\begin{equation}
    x_{k+1}=
    A x_{k}
    +
    B (u^L_k+u^F_k), \nonumber
\end{equation}
where $A=       \left[
                    \begin{array}{cccc}
                         1 & 0.5 & 0 & 0 \\
                         0 & 1 & 0 & 0 \\
                         0 & 0 & 1 & 0.5 \\
                         0 & 0 & 0 & 1
                    \end{array}
                \right]$,
    $B =        \left[
                    \begin{array}{cc}
                         0.125 & 0  \\
                         0.5 & 0  \\
                         0 & 0.125  \\
                         0 & 0.5 
                    \end{array}
                \right]$,
and state $x_k=[x, v_x, y, v_y]^T$ denotes $x$-position, $x$-velocity, $y$-position and $y$-velocity, the control input of the leader $u^L_k =[u^L_x, u^L_y]^T$  and the control input of the follower $u^F_k =[u^F_x, u^F_y]^T$ effect $x$-acceleration and $y$-acceleration together by their sum. The physical constraints are $x\in X = [0,10] \times [-3,3] \times [0,10] \times [-3,3]$, $u^L\in U^L = [-3,3]^2$ and $u^F\in U^F = [-0.01, 0.01]^2$.

{\bf Planning Objectives: }We assume that the initial state of the robot is $[2, 0, 6, 0]$, which is shown as the red point in Figure~\ref{fig:single-coop}. 
The control objective of the leader is to visit region A1 at least once within time interval 1 to 10 (from 0.5s to 5s), always stay at region A2 time interval  14 to 16 (from 7s to 8s), and finally reach region A3 at least once time interval 20 to 25 (from 10s to 12.5s). 
This objective can be specified by the STL formula 
\[
\phi^L=\mathbf{F}_{[1,10]}A_1 \land \mathbf{G}_{[14,16]}A_2 \land \mathbf{F}_{[20,25]}A_3,\]
where 
$A_1=(x\in [8,10])\land (y\in [8,10])$, $A_2=(x\in [1,4])\land (y\in [1,4])$ and $A_3= (x\in [8,10])\land (y\in [0,2])$. 
Meanwhile, the control objective of the follower is to visit region B1 at least once between instants 4 to 9 (from 2s to 4.5s) and always stay at region B2 within time interval 12 to 13 (from 6s to 6.5s), which is specified by $\phi^F=\mathbf{F}_{[4,9]}B_1 \land \mathbf{G}_{[12,13]}B_2$, where $B_1=(x\in [1,3])\land (y\in [6.5,8])$ and $B_2= (x\in [2,5])\land (y\in [2,5])$. The cost function we use is $J_S(\mathbf{x},\mathbf{u}^L,\mathbf{u}^F)= 1.6\times 10^{-7}\times\Sigma_{i=0}^{T_{\phi^L}}||u^L_i||_2^2-\rho^{\phi^L}(\xi_f(x_0,\mathbf{u}^L,\mathbf{u}^F),0)$. Note that we use the term $1.6\times 10^{-7}$ to match the order of magnitude of the cost of $\mathbf{u}^L$ and the robustness value of $\phi^L$.

\begin{figure}[t]
    \centering
    \includegraphics[width=0.4\textwidth]{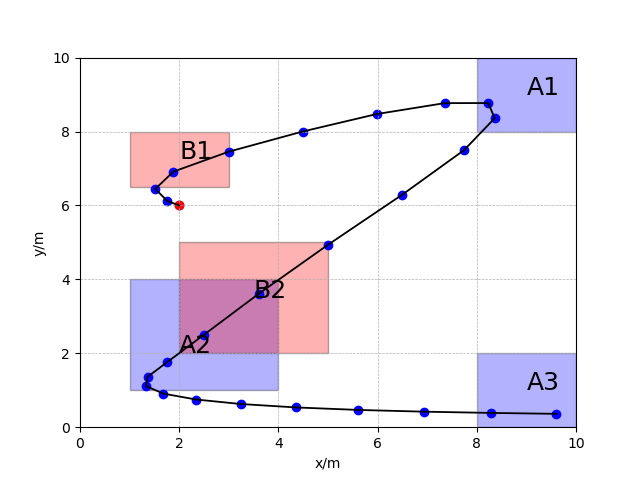}
    \caption{The trajectory for the single robot: cooperative case.}
    \label{fig:single-coop}
\end{figure}

\begin{figure}[t]
    \centering
    \includegraphics[width=0.4\textwidth]{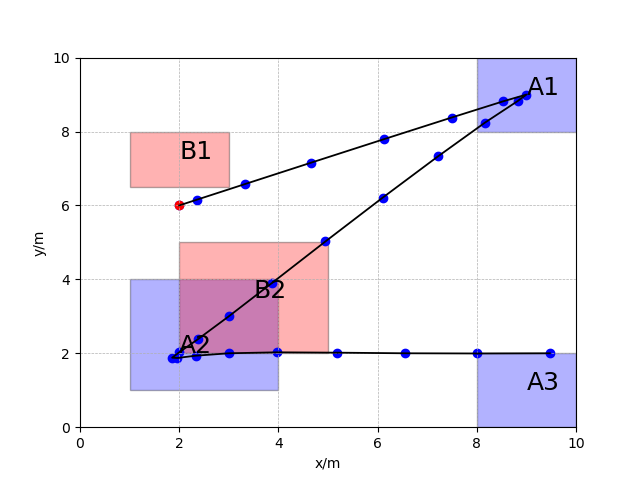}
    \caption{The trajectory for the single robot: antagonistic case.}
    \label{fig:single-anta}
\end{figure}

{\bf Simulation Results and Analysis: }
Note that, in this case study, the tasks $\phi^L$ and $\phi^F$ are not inherently conflicting. However, since the control input of the follower is restricted to the set $[-0.01, 0.01]^2$, the follower can only have a limited influence on the evolution of the system. As a result, whether the task $\phi^F$ is satisfied is primarily determined by the leader's actions. Specifically, the leader decides whether to adopt a cooperative solution or an antagonistic solution. This highlights the leader's dominant role in shaping the system's behavior in this scenario.

The resulting trajectories for the cooperative case and the antagonistic case are illustrated in Figure~\ref{fig:single-coop} and Figure~\ref{fig:single-anta}, respectively. The computation time and the associated cost for each case are summarized in Table~\ref{tab:1}. 
Note that the trajectories depicted in the figures are projections onto the first and third dimensions (representing the $x$-position and $y$-position, respectively). However, the set membership and satisfaction of the STL specifications are still evaluated in the complete 4-dimensional state space, which includes both position and velocity components.

In the cooperative case, the leader must first move to region B1 to collaborate with the follower. Then, it needs to adjust the speed to ensure the entire system remains in region B2 during the time interval $[12, 13]$. As a result, the cooperative approach leads to (i) higher control cost and (ii) lower robustness of the formula $\phi^L$ compared to the antagonistic case. In the antagonistic case, since the follower cannot achieve $\phi^F$, it chooses to be non-interfering, and the computation time for this case is much shorter than for the cooperative case.

\begin{table}[t]
    \caption{}
    \centering
    \begin{tabular}{c|c|c}
    \hline
     Procedure & Computation Time & Cost\\    
    \hline
     Cooperative & 7.4825s & -0.3613\\
    \hline
     Antagonistic & 0.6691s & -0.9999\\
    \hline
    
    \end{tabular}
    \vspace{10pt}
    \label{tab:1}
\end{table}

\subsection{Multi-Agent Planning Problem}\label{sec-multi}
{\bf System Model: }
In this case study, we consider a team of three identical robots  in a common workspace. 
The dynamic of each robot is modeled by a single integrator on a planar workspace: 
\begin{equation}
    x^i_{t+1}= A x^i_t + B u^i_t \nonumber
\end{equation}
where $A=       \left[
                    \begin{array}{cccc}
                         1 & 0 \\
                         0 & 1 
                    \end{array}
                \right]$,
    $B =        \left[
                    \begin{array}{cc}
                         1 & 0  \\
                         0 & 1 
                    \end{array}
                \right],i=1, 2, 3$,  
and state $x_t^i=[x^i, y^i]^T$ denotes $x$-positions and $y$-positions for agent $i$.  
The physical constraints are $x^i\in X^i=[0,10]^2$, $u^i\in U^i=[-1,1]^2$ for all $i= 1,2,3$. 
All  robots start from the same initial state $(2,6)$.

{\bf Planning Objectives: } 
Our design objective is to ensure that all three agents can achieve their respective tasks. However, we assume that the designer can only directly control the behavior of Agent 1, while Agents 2 and 3 are rational agents that jointly pursue their own objectives. Consequently, the behavior of Agents 2 and 3 is implicitly influenced by finding a suitable cooperative solution for Agent 1, which aligns their actions with the overall system goals.  

The control objective of Agent 1 is given by the following STL formula 
\[ 
\phi^L=\mathbf{F}_{[1,10]}A_1 \land \mathbf{G}_{[14,16]}A_2 \land \mathbf{F}_{[20,25]}A_3,
\]
where $A_1=(x^1\in [0,2])\land (y^1\in [8,10])$, $A_2=(x^1\in [7,10])\land (y^1\in [7,10])$, and $A_3= (x^1\in [8,10])\land (y^1\in [0,2])$.
The control objectives of Agents 2 and 3 are as follows. First, during the entire operation horizon $[0,25]$, both Agents 2 and 3 are required to follow the trajectory of Agent 1 by maintaining a distance of less than one unit. This requirement can be formally specified using the following STL formulae:
\[
\phi_D^i=\mathbf{G}_{[0,25]} (x^1-x^i)^2+(y^1-y^i)^2\leq 1, i=2,3.
\] 
In addition, each follower has its own task. Specifically, Agent 2 is required to visit region $B_1$ at least once within the time interval $[4,9]$ and remain within region $B_2$ during the time interval $[15,17]$. These requirements can be formally specified using the following STL formula: 
$\phi^2=\mathbf{F}_{[4,9]}B_1 \land \mathbf{G}_{[15,17]}B_2$, where $B_1=(x^2\in [1,3])\land (y^2\in [6.5,8])$ and $B_2= (x^2\in [8,10])\land (y^2\in [6,9]$). 
For Agent 3, it needs to  visit region $C_1$ at least once within the time interval $[1,25]$, which is specified by STL formula $\phi^3=\mathbf{F}_{[1,25]}C_1$, where $C_1=(x^3\in [3,6])\land (y^3\in [3,6])$. 
In summary, the overall STL formula for the followers is given by  
\[
\phi^F=\phi_D^2\land \phi_D^3 \land \phi^2 \land \phi^3.
\]
In this experiment, we inherit the cost function used in the above experiment $J_S(\mathbf{x},\mathbf{u}^1,\mathbf{u}^2, \mathbf{u}^3)= 1.6\times 10^{-7}\times\Sigma_{i=0}^{T_{\phi^L}}||u^L_i||_2^2-\rho^{\phi^L}(\xi_f(x_0,\mathbf{u}^1,\mathbf{u}^2, \mathbf{u}^3),0)$. 
We note that here since we consider a single integrator system, 
the robustness of $\phi^L$ generally dominates the value of $J_S$.

% We then illustrate the applications of our approach to the multi-agent planning problem. To formalize the behaviors of the follower agent, we consider the following setting:
% \begin{itemize}
%     \item The follower determines $\mathbf{u}^F$ by achieving objective $\phi^F$ at first, i.e., $\mathbf{u}^F\in \mathbf{SR}(\mathbf{u}^L,{\phi^F})$ when $\mathbf{SR}(\mathbf{u}^L,{\phi^F})\ne\emptyset$;
%     \item When $\phi^F$ is not achievable, the follower determines $\mathbf{u}^F$ by achieving an alternative objective $\phi^A$ instead, i.e., $\mathbf{u}^F\in \mathbf{SR}(\mathbf{u}^L,{\phi^A})$ when $\mathbf{SR}(\mathbf{u}^L,{\phi^F})=\emptyset$;
%     \item The achievability of alternative objective $\phi^A$ is independent of $\mathbf{u}^L$, i.e., $\mathbf{SR}(\mathbf{u}^L,{\phi^A})\ne\emptyset$ for all $\mathbf{u}^L\in (U^F)^L$;
%     \item The leader knows both $\phi^F$ and $\phi^A$.
% \end{itemize}

% To solve this problem, here we choose to incorporate Algorithm~\ref{algorithm-synthesis-coop} to solve the antagonistic case. Specifically, we (i) remove constraint~(\ref{cons-7b}); (ii) change $z^F$ for $z^A$ in Equation~(\ref{equation-7}) and Algorithm~\ref{algorithm-synthesis-coop}; (iii) add constraints~(\ref{cons-9b}) to Equation~(\ref{equation-7}); and (iv) add the condition $\vee(z^F=1)$ for determining counterexamples. 

\begin{figure}[t]
    \centering
    \includegraphics[width=0.4\textwidth]{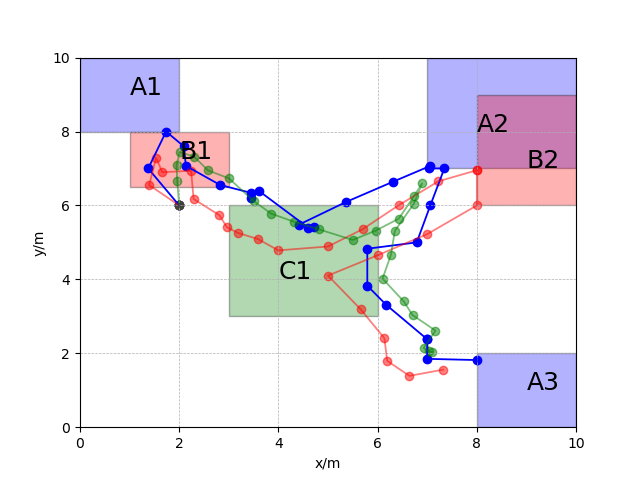}
    \caption{Trajectories for three robots: cooperative case.}
    \label{fig:multi-coop}
\end{figure}

% \begin{figure}[t]
%     \centering
%     \includegraphics[width=0.4\textwidth]{Figure_multi_anta.png}
%     \caption{Trajectory for the leader robot: antagonistic case.}
%     \label{fig:multi-anta}
% \end{figure}

{\bf Simulation Results and Analysis: }
The simulation result of the cooperative solution is shown in Figure~\ref{fig:multi-coop}, where  blue, red and green  lines denote the trajectories of Agents 1, 2 and 3, respectively. 
The computation time for finding the solution is $11.4493$s and the value of the optimal cost is $2.7439e-6$.

Note that, while tasks \(\phi^L\) and \(\phi^F\) do not directly conflict, the mobility of Agents 2 and 3 is constrained by the trajectory of Agent 1, as they need to closely follow the leader's path. Therefore, when synthesizing the trajectory for Agent 1, it cannot pursue its own STL task directly, as doing so would prevent the followers from achieving their reachability tasks while maintaining a close distance. 
For example, to collaborate with Agent 2, the leader (Agent 1) must first stay close to region \(B_1\) at the start of the execution, and then remain close to region \(B_2\) during the time interval \([15, 17]\). While these constraints are not necessary for Agent 1 itself and would decrease the robustness of its task \(\phi^L\), such sacrifices are essential to ensure the feasibility of the tasks for Agents 2 and 3.

 % On the contrary, in the antagonistic case, the leader directly goes through region B1 rapidly to visit A1, and goes to region A2 along the boundary of the environment to avoid collaborating with robot 3. Such a policy ensures the robustness of $\phi^L$ since robot 1 can visit the central point of region A1, A2, and A3, but it also leads to a higher control cost.

\section{Conclusion}\label{sec-conclusion}
In this paper, we present a Stackelberg game-theoretical framework for signal temporal logic  control synthesis in the presence of uncontrollable agents. Compared to existing robust control approaches, our framework better captures the rationality of uncontrollable agents, providing a more flexible and less conservative solution for the system controller. This work represents the first step toward STL control synthesis within the Stackelberg game framework. However, the current results have certain limitations that we aim to address in future research. First, although we consider a leader-follower setting, each agent still operates in an open-loop control fashion, i.e.,  it commits to a fixed plan at the beginning and execute it without adaptation. In the future, we plan to extend this to a reactive control setting, where agents can adjust their decisions on-the-fly based on real-time feedback. Additionally, our current framework focuses on the Boolean satisfaction of STL formulae. We intend to expand our results to a quantitative setting, where, for example, the follower maximizes the robustness value of its STL formula rather than abandoning its task when qualitative satisfaction is not achievable. 

\bibliographystyle{ieeetr}
\bibliography{mybib}

\end{document}